\documentclass[10pt,journal,english,twocolumn]{IEEEtran}
\usepackage{blindtext}

\usepackage{babel}
\usepackage[babel]{microtype}
\usepackage[T1]{fontenc}

\usepackage{graphicx}
\usepackage{xcolor}
\usepackage{tikz}
\usepackage{pgfplots}

\usepackage{tabularx}
\usepackage{booktabs}
\usepackage{stfloats}

\usepackage{amsmath}
\usepackage{amsfonts}
\usepackage{amssymb}
\usepackage{amsthm}
\usepackage{bm}
\usepackage{siunitx}

\usepackage{csquotes}
\usepackage[backend=biber,url=false,style=ieee,isbn=false]{biblatex}
\bibliography{literature.bib}
\renewcommand*{\bibfont}{\footnotesize}

\usepackage{hyperref}
\hypersetup{
	pdfauthor={Karl-Ludwig Besser},
	pdftitle={Reliability Bounds for Dependent Fading Wireless Channels},
	colorlinks=true,
	allcolors=.,
	urlcolor=blue,
}

\usepackage[acronyms,nomain,xindy]{glossaries}
\makeglossaries
\glsenableentrycount
\newacronym{5g}{5G}{Fifth-Generation Wireless Systems}
\newacronym{ask}{ASK}{amplitude-shift keying}
\newacronym[plural={AWGN}, \glsshortpluralkey={AWGN}]{awgn}{AWGN}{additive white Gaussian noise}

\newacronym{ber}{BER}{bit error rate}
\newacronym{bler}{BLER}{block error rate}
\newacronym{bpsk}{BPSK}{binary phase-shift keying}

\newacronym{cdf}{CDF}{cumulative distribution function}
\newacronym{cnn}{CNN}{convolutional neural network}
\newacronym{csi}{CSI}{channel-state information}
\newacronym{csir}{CSI-R}{channel-state information at the receiver}
\newacronym{csit}{CSI-T}{channel-state information at the transmitter}

\newacronym{ecc}{ECC}{error-correcting code}

\newacronym{ff}{FF}{feed-forward}
\newacronym{ffnn}{FF-NN}{feed-forward neural network}
\newacronym{fsk}{FSK}{frequency-shift keying}

\newacronym{gm}{GM}{Gaussian mixture}

\newacronym{iid}{i.\,i.\,d.}{independent and identically distributed}
\newacronym{il}{IL}{information leakage}
\newacronym{iot}{IoT}{internet of things}

\newacronym{lb}{LB}{lower bound}

\newacronym{map}{MAP}{maximum a posteriori}
\newacronym{ml}{ML}{machine learning}
\newacronym{mop}{MOP}{multi-objective programming problem}
\newacronym{mse}{MSE}{mean squared error}

\newacronym{nn}{NN}{neural network}
\newacronym{nve}{NVE}{normalized validation error}

\newacronym{pdf}{PDF}{probability density function}
\newacronym{pwc}{PWC}{polar wiretap code}

\newacronym{qam}{QAM}{quadrature amplitude modulation}

\newacronym{rbf}{RBF}{radial basis function}
\newacronym{relu}{ReLU}{rectified linear unit}
\newacronym{rnn}{RNN}{recurrent neural network}

\newacronym{sgd}{SGD}{stochastic gradient descent}
\newacronym{sgp}{SGP}{secure goodput}
\newacronym{snr}{SNR}{signal-to-noise ratio}
\newacronym{svm}{SVM}{support vector machine}

\newacronym{tvd}{TVD}{total variation distance}

\newacronym{ub}{UB}{upper bound}
\newacronym{urllc}{URLLC}{ultra-reliable low latency communication}
\setacronymstyle{long-short}

\pgfplotsset{compat=newest}
\pgfplotsset{plot coordinates/math parser=false}

\usetikzlibrary{shapes, arrows, circuits, trees, positioning, fit, external}

\newcommand{\diff}{\ensuremath{\mathrm{d}}}
\newcommand{\expect}[2][]{\ensuremath{\mathbb{E}_{#1}\left[#2\right]}}
\newcommand{\inv}[1]{\ensuremath{#1^{-1}}}
\newcommand{\abs}[1]{\ensuremath{\left|#1\right|}}

\newcommand{\worstout}{\ensuremath{\underline{R}^{\varepsilon}}}
\newcommand{\bestout}{\ensuremath{\overline{R^{\varepsilon}}}}

\renewcommand{\Pr}{\ensuremath{\mathbb{P}}}
\theoremstyle{plain}%
\newtheorem{thm}{Theorem}%
\newtheorem{lem}[thm]{Lemma}
\newtheorem{prop}[thm]{Proposition}
\newtheorem{cor}{Corollary}

\theoremstyle{remark}
\newtheorem{rem}{Remark}
\newtheorem{example}{Example}

\definecolor{plot0}{HTML}{2db7d2}
\definecolor{plot1}{HTML}{23373b}
\definecolor{plot2}{HTML}{EB811B}
\definecolor{plot3}{HTML}{14B03D}

\title{Reliability Bounds for\\Dependent Fading Wireless Channels}
\author{Karl-Ludwig Besser, \IEEEmembership{Student Member, IEEE} and Eduard A. Jorswieck, \IEEEmembership{Fellow, IEEE}%
\thanks{The authors are with the Institute of Communications Technology, Technische Universit\"at Braunschweig, 38016 Braunschweig, Germany (email: k.besser@tu-bs.de; e.jorswieck@tu-bs.de).}
\thanks{This work is supported in part by the German Research Foundation (DFG) under grant JO\,801/23-1.}
}

\newcommand\copyrighttext{%
	\scriptsize \copyright 2020 IEEE. Personal use of this material is permitted. Permission from IEEE must be obtained for all other uses, in any current or future media, including reprinting/republishing this material for advertising or promotional purposes, creating new collective works, for resale or redistribution to servers or lists, or reuse of any copyrighted component of this work in other works.}
\newcommand\copyrightnotice{%
	\begin{tikzpicture}[remember picture,overlay]
	\node[anchor=south,yshift=5pt] at (current page.south) {\fbox{\parbox{\dimexpr\textwidth-\fboxsep-\fboxrule\relax}{\copyrighttext}}};
	\end{tikzpicture}%
}

\begin{document}
\maketitle

\begin{abstract}
	Unreliable fading wireless channels are the main challenge for strict performance guarantees in mobile communications. Diversity schemes including massive number of antennas, huge spectrum bands and multi-connectivity links are applied to improve the outage performance. The success of these approaches relies heavily on the joint distribution of the underlying fading channels. In this work, we consider the $\varepsilon$-outage capacity of slowly fading wireless diversity channels and provide lower and upper bounds for fixed marginal distributions of the individual channels. This answers the question about the best and worst case outage probability achievable over $n$ fading channels with a given distribution, e.g., Rayleigh fading, but not necessarily statistically independent. Interestingly, the best-case joint distribution enables achieving a zero-outage capacity greater than zero without channel state information at the transmitter for $n \geq 2$. Furthermore, the results are applied to characterize the worst- and best-case joint distribution for zero-outage capacity with perfect channel state information everywhere.
	All results are specialized to Rayleigh fading and compared to the standard assumption of independent and identically distributed fading component channels. The results show a significant impact of the joint distribution and the gap between worst- and best-case can be arbitrarily large.
\end{abstract}

\copyrightnotice

\section{Introduction}
With advances in communication technology, more critical applications start to rely on wireless transmission, e.g., car-to-car communication and medical applications~\cite{Andrews2014}. These areas make high demands on reliability. Therefore, research started to focus on problems like ultra-reliable communications where very low error-rates of less than $10^{-3}$ are required~\cite{Bennis2018}. In order to understand the trade-offs and efficient operating points for ultra-reliable communications, \cite{Bennis2018} develops a framework by listing enabling technologies and methods as well as their application in the use cases 1) enhanced Mobile Broadband (eMBB), 2) massive Machine-Type Communication (mMTC), and 3) \gls{urllc}~\cite{ITU-T}.

In mobile wireless settings, the communication channel is usually modeled as a (slow) fading channel~\cite{Biglieri1998}. In this type of channel, one cannot transmit code words with an arbitrarily small error probability even for an infinite blocklength~\cite[Chap.~5]{Tse2005}. Therefore, the $\varepsilon$-outage capacity is used as a performance metric. It is defined as the largest transmission rate for which the outage probability is still less than $\varepsilon$~\cite{Tse2005}.

Various techniques to lower the outage probability or to increase the  diversity order are developed \cite{Wornell1998} including all available dimensions, temporal diversity, spectral diversity \cite{Nguyen2019}, spatial and multi-connectivity diversity \cite{Simsek2019}. The underlying idea is simple: if one symbol or codeword travels along several fading paths to the receiver, the probability that it is received correctly is increased. However, this conclusion depends heavily on the underlying joint probability density. 

In order to motivate the work in this paper, consider the situation where we have a standard probability space $(\Omega, \mathcal{F}, \mathbb{P})$ with $n$ random variables $X_1,...,X_n$, each with fixed univariate marginal distribution $X_{i} \sim {F}_i$, $i=1,...,n$. \emph{Question 1: What can we say about the distribution of $S=X_1+...+X_n$ when the dependence structure among the random variables is arbitrary?}

In our reliable communication over fading channels scenario, this corresponds to the situation where we have collected many fading channel gain $X_i$ measurements at $n$ specific points in space and/or frequency, from which we can deduce the marginal distribution ${F}_i$ of the fading channels at these points. \emph{Question 2: What can we say about the achievable outage performance $\mathbb{P}(S<r)$ for $r>0$ and $S=X_1+...+X_n$, if we employ a receive diversity system at these points and, e.g., perform coherent combining?}

Obviously both Questions 1 and 2 above are closely related. Let us illustrate the importance of these questions with a simple dichotomy example from \cite{Jorswieck2019} with $n=2$ and Rayleigh fading channels, i.e., standard exponentially distributed random variables $X_1$ and $X_2$ for the so-called zero-outage capacity: if $X_1$ and $X_2$ are completely dependent, i.e., $X_1=X_2$ with probability one, then the zero-outage capacity is zero, while if $X_1$ and $X_2$ are completely negatively dependent, the zero-outage capacity is strictly larger than zero.  

For frequency diversity, \cite{Haber1974} and later \cite{Akki1985}, have shown that it is possible to design negatively correlated branches by selecting frequency differences properly. Indeed, average error probabilities for \gls{fsk} receiver and square-law combiner \cite{Haber1974}, coherent and non-coherent \gls{fsk}, and differential and coherent PSK with maximum ratio combining \cite{Akki1985} are shown to vanish at certain frequency spacings. In \cite{Hamdi2008}, it is observed that negatively correlated branches in Rician fading can lead to an increase in the ergodic capacity compared to the uncorrelated case. Different diversity combiner schemes in lognormal fading channels are analyzed in \cite{Zhu2018} and it is observed that negatively correlated lognormal channels can outperform the independent channels. 

From the theoretical side, the impact of the joint distribution on different performance metrics is studied in the context of entropy in \cite{Cover1994} and for performance bounds in \cite{Biglieri2016}. The problem of optimizing the joint distribution given marginal distributions can be formulated by the copula\footnote{A copula is basically a mapping from a set of marginal distributions to a joint distribution.} approach \cite{Nelsen2006}. Recently, \cite{Lin2017} exploits copulas and coupling to derive stochastic orders for multi-user fading channels to characterize ergodic capacity regions of multi-user channels \cite{Lin2019}.

However, most of the previous work on performance of diversity systems in fading wireless channels only considers the case of independent or positively correlated fading processes \cite{Alouini2002}. In this work, we will partly answer the Question 2 above and provide bounds for cases which are relevant for the fading wireless channels. However, please note that the related Question 1 is still open in general, even for the case $n=2$ \cite{Wang2016}. 

In this work, we adapt fundamental results from operations research to the reliability of wireless communication links and based on this transfer, it is possible to obtain the characterization of worst- and best-case fading correlations in diversity systems. We provide both upper and lower bounds on various performance metrics for fading wireless channels with monotone marginal densities allowing dependent fading processes. The performance metrics include $\varepsilon$-outage capacity and, as a special case, the zero-outage capacity.
The main contributions are summarized in the following.
\begin{itemize}
	\item In Theorem~\ref{thm:bounds-outage-identical}, we provide lower and upper bounds for fading channels where all fading coefficients $\abs{h_i}^2\sim F$ follow the same distribution with a monotone density. The bounds are tight and are achieved by a certain joint distribution.
	\item We show the surprising results that the best-case distribution achieves strictly positive values for all monotone fading gain distributions.
	\item Theorem~\ref{thm:bounds-outage-different} extends the result from Theorem~\ref{thm:bounds-outage-identical} to channels with arbitrary and different monotone marginal densities $\abs{h_i}^2\sim F_i$. In general, these bounds are not tight. 
	\item For perfect \gls{csit}, the best and worst-case are shown where the worst case corresponds to completely dependent, i.e., comonotonic fading channels. 
	\item The derived results are applied to the typical Rayleigh fading model and the bounds are stated explicitly, including channels with different large-scale fading. 
	\item From the numerical assessment, the two most interesting observations are 1) the large impact of the joint distribution on the $\varepsilon$-outage capacity and 2) for the best-case joint distribution, the zero-outage capacity, with and without \gls{csit}, converges with increasing diversity dimension $n$. 
\end{itemize}

The paper is organized as follows. In Section~\ref{sec:preliminaries}, preliminaries and the system model are introduced. The main results of bounds on the $\varepsilon$-outage capacity with only \gls{csir} are presented in Section~\ref{sec:bounds}. Results for the zero-outage capacity with additional \gls{csit} are shown in Section~\ref{sec:perfect-csit}. An example of the application of all results for the special case of Rayleigh fading channels is given in Section~\ref{sec:rayleigh} including numerical assessment of the impact of different system parameters on the bounds and the state of the art, mainly the \gls{iid} case. Section~\ref{sec:conclusion} concludes the paper.

\subsection*{Notation}
Throughout this work, we will use $F$ and $f$ for a probability distribution and its density, respectively. The expected value is denoted by $\mathbb{E}$ and the probability by $\Pr$. The function $G=\inv{F}$ denotes the inverse of $F$. It is assumed that all considered distributions are continuous. Real-valued Gaussian and complex circularly symmetric Gaussian random variables with mean $\bm{\mu}$ and covariance matrix $\Sigma$ are denoted as $\mathcal{N}(\bm{\mu}, \Sigma)$ and $\mathcal{CN}(\bm{\mu}, \Sigma)$, respectively. The uniform distribution on the interval $[a, b]$ is denoted as $\mathcal{U}[a, b]$.

\section{Preliminaries and System Model}\label{sec:preliminaries}

\subsection{System Model}
Throughout this work, we consider the complex flat fading channel~\cite[Ch.~5.4]{Tse2005} with $n$ receive dimensions. These could be $n$ antennas placed spatially or $n$ distributed receivers like in multi-connectivity, or $n$ time or frequency instances, over which the same symbol $x$ is transmitted. The received signal at discrete time $k$ is given by the vector $\bm{y}$ with $n$ components as
\begin{equation}\label{eq:signal-model}
\bm{y}[k] = \bm{h}[k]x[k] + \bm{w}[k],
\end{equation}
where $\bm{h}[k]=[h_1[k], \dots{}, h_n[k]]$ represents the fading channel and $\bm{w}[k]\sim\mathcal{CN}(0, N_0)$ is \gls{iid} complex Gaussian noise with zero mean and variance $N_0$. In the following, we will drop the time instance $k$ since we assume that $\bm{h}$ is constant for the code word length and the corresponding rate expressions are achievable.

If the transmitter has no \gls{csit}, and the receiver has perfect \gls{csi} on the slow-fading channel, then the definition of the $\varepsilon$-outage capacity $R^{\varepsilon}$ for this channel model is given as \cite{Tse2005}
\begin{equation}\label{eq:def-outage}
R^{\varepsilon} = \sup_{R\geq 0}\left\{R \in \mathbb{R}: \Pr\left(\sum_{i=1}^{n} \abs{h_i}^2 <\frac{2^R-1}{\rho}\right)<\varepsilon\right\}
\end{equation}
for a certain \gls{snr} $\rho$. The \gls{snr} of the channel is defined as $P/N_0$ where $P$ is the transmit power. The $\varepsilon$-outage capacity is the transmission rate which results in an outage probability of at most $\varepsilon$, if the next channel realization is used. 

One extreme case of $\varepsilon$-outage capacity is the zero-outage capacity \cite{Li2001a} or delay-limited capacity \cite{Hanly1998}, where $\varepsilon$ is set to zero. This means the rate $R^0$ is achievable \emph{for all channel realizations}. It is well known that without \gls{csit}, this is not possible (at least for \gls{iid} fading with the non-negative real numbers as support), while with perfect \gls{csit} it can be achieved with a long-term power constraint and channel inversion \cite{Goldsmith1997}. More general, it is known that the zero-outage capacity only depends on the support of the joint distribution~\cite{Verdu1994}.
	
\subsection{Problem Statements}	
	
In this work, we are interested in the following problem statements: 
\begin{description}
	\item[P1] For no instantaneous \gls{csit} and perfect \gls{csir}, and for fixed and known marginal distributions of the channel gains $|h_1|^2, ..., |h_n|^2$, what is the worst-case and best-case $\varepsilon$-outage capacity over all possible joint distributions? 
	\item[P2] For perfect \gls{csi} at both transmitter and receiver and long-term power constraint, and for fixed marginal distributions of the channel gains, what is the worst-case and best-case zero-outage capacity over all possible joint distributions? 
\end{description}	
As a follow up question, we are also interested in the comparison between the two cases of zero-outage capacity. 

\subsection{Bounds on the Outage Performance or Risk}	
	
In many different areas, the sum of multiple random variables with (unspecified) dependency structure is of interest. An example is the total risk in risk management \cite{Embrechts2010}. Since the dependency between the different variables is usually unknown, bounds like the worst and the best case are of particular interest. In \cite{Wang2013}, the authors provide bounds on the probability of the sum of dependent variables $X_i\sim F_i$ with monotone marginals
\begin{equation}
m(s) = \inf_{F_{X_1, \dots{}, X_n}}\left\{\Pr\left(\sum_{i=1}^{n}X_i < s\right): X_i\sim F_i\right\}\,,
\end{equation}
where a distribution $F_i$ is referred to as monotone marginal, if its density is monotone on its support.
In \cite{Wang2013}, a bound for $m(s)$ is provided, which we restate in the following theorem.
\begin{thm}[{\cite[Thm.~2.6.1]{Wang2013}}]
	
	Suppose the distributions $F_1, \dots{}, F_n$ are continuous, then the following holds
	\begin{equation}\label{eq:relation-Phi-m}
		m(s) \geq \inv{\Phi}(s)\,,
	\end{equation}
	where $\Phi$ is the conditional moment function, given as 
	\begin{equation}\label{eq:def-Phi}
	\Phi(a) = \sum_{i=1}^{n}\expect[X_i\sim F_i]{X_i | X_i \geq G_i(a)}\,,
	\end{equation}
	with $G_i=\inv{F_i}$ being the inverse of the \gls{cdf} $F_i$.
\end{thm}

For the homogeneous case $F_1 = \cdots{} = F_n = F$ with monotone density, the authors of \cite{Wang2013} provide a formulation of a function $\phi$ such that \eqref{eq:relation-Phi-m} is fulfilled with equality. In order to do this, they prove the existence of an optimal coupling between the dependent variables by a joint distribution $Q_{n}^{F}$.
In order to construct this joint distribution and the function $\phi$, we need the following functions $c_n$ and $H$.
For decreasing densities, $c_{n}$ and $H$ are defined as \eqref{eq:cmin-decreasing}
\begin{figure*}[b]
	\noindent\rule{\textwidth}{.5pt}
	\begin{equation}
	\label{eq:cmin-decreasing}%
	c_{n}(a) = \min\left\{c\in\left[0, \frac{1-a}{n}\right]: \int_{c}^{\frac{1-a}{n}} H_{a}(t)\diff{t} \geq \left(\frac{1-a}{n} - c\right) H_{a}(c)\right\}
	\end{equation}
\end{figure*}
and
\begin{equation}\label{eq:h-decreasing}
	H_{a}(x) = (n-1) G(a+(n-1)x) + G\left(1-x\right),
\end{equation}
with $a\in\left[0, 1\right]$~\cite{Wang2013}.

For increasing densities, $c_{n}$ and $H$ are given as \eqref{eq:cmin-increasing}
\begin{figure*}[t]
	\begin{equation}
	\label{eq:cmin-increasing}
	c_{n}(a) = \min\left\{c\in\left[0, \frac{1-a}{n}\right]: \int_{c}^{\frac{1-a}{n}} H_{a}(t)\diff{t} \leq \left(\frac{1-a}{n} - c\right) H_{a}(c)\right\}
	\end{equation}
\end{figure*}
and
\begin{equation}\label{eq:h-increasing}
	H_{a}(x) = G(a+x) + (n-1)G\left(1-(n-1)x\right),
\end{equation}
with $a\in\left[0, 1\right]$~\cite{Wang2013}.

Combining \eqref{eq:cmin-decreasing}-\eqref{eq:h-increasing} into one function $\phi$, gives the following definitions for decreasing densities
\begin{equation}
	\label{eq:phi-decreasing}
	\phi(a) = \begin{cases}
	H_{a}(c_{n}(a)) & \text{if } c_{n}(a) > 0\\
	n\expect{X|X>G(a)} & \text{if } c_{n}(a) = 0
	\end{cases},
\end{equation}
and increasing densities
\begin{equation}
	\label{eq:phi-increasing}
	\phi(a) = \begin{cases}
	H_{a}(0) & \text{if } c_{n}(a) > 0\\
	n\expect{X|X>G(a)} & \text{if } c_{n}(a) = 0
	\end{cases}.
\end{equation}

As mentioned earlier, the inverse of $\phi$ gives the tight lower bound on the probability of the sum of the random variables. We restate this central result from \cite{Wang2013} in the following theorem. Please refer to \cite{Wang2011,Wang2013} for the proofs and further details.
\begin{thm}[{\cite[Thm.~3.4]{Wang2013}}]
	Suppose the distribution function $F$ has a decreasing density on its support and $\phi(a)$ is defined in \eqref{eq:phi-decreasing}, or $F$ has an increasing density on its support and $\phi(a)$ is defined in \eqref{eq:phi-increasing}. Then
	\begin{equation}\label{eq:phi-equal-m}
		\inv{\phi}(s) = m(s)%
	\end{equation}
	holds.
\end{thm}

The important result in \eqref{eq:phi-equal-m} characterizes the worst-case joint distribution in computable, almost closed-form. 

In~\cite[Sect.~3]{Wang2011}, a specific joint distribution $Q_n^F(c)$ for $n \geq 2$ is constructed for some $0 \leq c \leq \frac{1}{n}$, where $(U_1,\dots{},U_n) \sim Q_n^F(c)$ satisfy one of the following two cases:
\begin{enumerate}
\item For each $i=1,...,n$, the joint density of $(U_1,\dots{},U_n)$ given $U_i \in [0,c]$ is uniformly supported on line segments $u_j=1-(n-1)u_i$, for all $j \neq i$, and 
\item $G_1(U_1)+\dots{}+G_n(U_n)$ is a constant when $U_i \in (c,1-(n-1)c)$ for any $i=1,...,n$.
\end{enumerate}
The intuition behind this construction is that for realizations of $U_i$ which lie in the medium area $(c,1-(n-1)c)$ a complete mix is created for which the sum of the realizations is constant. While outlier realizations $U_i \in [0,c]$ close to zero, are compensated by the other $U_j$, $j \neq i$.  
In~\cite[Thm.~3.4]{Wang2011}, the authors use this joint distribution $Q_{n}^{F}$ to solve the following minimization problem.

\begin{thm}[{\cite[Thm.~3.4]{Wang2011}}]
	Suppose $F$ is a distribution function with increasing density. Then for any convex function $g$
	\begin{equation}\label{eq:minimization-problem-wang}
	\min_{X_1, \dots{}, X_n\sim F}\expect{g(X_1+\cdots{}+X_n)} = \expect[Q_{n}^{F}]{g(X_1 + \cdots{}+X_n)}
	\end{equation}
	holds.
\end{thm}

If $X\sim F$ has a decreasing density $f(x)$, $-X$ has an increasing density. In the following, $\phi_{-}$ denotes the function $\phi$ from \eqref{eq:phi-decreasing} and \eqref{eq:phi-increasing} for the random variable $-X\sim F$ with the density $f(-x)$. Note that if $\phi$ is defined according to \eqref{eq:phi-decreasing}, $\phi_{-}$ will be evaluated according to \eqref{eq:phi-increasing}.
More details on the calculations and relations between the different functions are listed in the appendix.

\section{\texorpdfstring{Bounds on the $\varepsilon$-Capacity for Monotone Marginals and Perfect CSI-R}{Bounds on the Epsilon-Capacity for Monotone Marginals and Perfect CSI-R}}
\label{sec:bounds}
In this section, upper and lower bounds for fading channels with fixed marginals and possibly dependent fading coefficients are derived. Throughout this section, only perfect \gls{csir} is assumed. By this we are able to answer problem statement \textbf{P1} completely. Let us stress that the solution builds on, but also extends the results from operations research reviewed in Section \ref{sec:preliminaries} significantly, because we are able to consider the typical properties of our fading channel gain distributions. Therefore, we can characterize the upper and lower bounds easier and more efficiently compute interesting asymptotic cases.

Throughout this section, we will give short examples illustrating how the derived theorems can be used to actually calculate the bounds on the $\varepsilon$-outage capacity. You can also find all calculations in detail and with interactive plots at \cite{BesserGitlab}.

\subsection{Identical Marginals}
\label{sub:identical-marginals}
At first, we will derive an upper and a lower bound on the outage capacity for the scenario where all fading coefficients $h_i$ have the same marginal distribution, i.e., $\abs{h_i}^2\sim F, i=1, \dots{}n$. 

\begin{thm}[Bounds on the $\varepsilon$-outage capacity for identical monotone marginals]\label{thm:bounds-outage-identical}
	The $\varepsilon$-capacity $R^{\varepsilon}$ of $n$ multi-connectivity fading links with monotone marginal distributions $\abs{h_i}^2\sim F$ and perfect \gls{csir} can be bounded by
	\begin{equation}
		\worstout \leq R^{\varepsilon} \leq \bestout \,.
	\end{equation}
	The worst-case $\varepsilon$-capacity $\worstout$ is given as
	\begin{equation}\label{eq:worst-outage}
		\worstout(\rho) = \log_2\left(1-\rho \cdot \phi_{-}(1-\varepsilon)\right),
	\end{equation}
	and the best-case $\varepsilon$-capacity $\bestout$ is given as
	\begin{equation}\label{eq:best-outage}
		\bestout(\rho) = \log_2\left(1 + \rho \cdot \phi(\varepsilon)\right),
	\end{equation}
	where $\rho$ is the \gls{snr} and $\phi$ is defined in Eq.~\eqref{eq:phi-decreasing} and \eqref{eq:phi-increasing}.
\end{thm}
\begin{proof}
	The proof consists of two parts in which the lower and upper bound are derived. 
	
	1) At first, we will prove the lower bound.
	For the outage-capacity defined in \eqref{eq:def-outage}, the upper bound is determined by the worst-case\footnote{Note that the worst-case outage probability is the supremum while the best-case outage probability corresponds to the infimum in $m(s)$.} probability,
	\begin{equation*}
		M(s) = \sup\{\Pr(S < s): X_i\sim F\},
	\end{equation*}
	with $X_i=\abs{h_i}^2$, $S=\sum_{i=1}^{n} X_i$, and $s=\frac{2^R-1}{\rho}$.
	
	This yields the following expression
	\begin{align}
		s^{\star}(\worstout) &= \sup_{s\geq 0}\left\{s: M(s) < \varepsilon\right\}\\
		&= \sup_{s\geq 0}\left\{s: 1-\inf\{\Pr(-S < -s)\}<\varepsilon\right\}\\
		&= \sup_{s\geq 0}\left\{s: 1-\inv{\phi_{-}}(-s)<\varepsilon\right\}
	\end{align}
	where $\inv{\phi_{-}}$ denotes the inverse of the function $\phi$ from \eqref{eq:phi-decreasing} and \eqref{eq:phi-increasing} for $-X_i\sim F$. The last transformation is satisfied by \eqref{eq:phi-equal-m}.	%
	The supremum is attained for equality
	\begin{equation*}
		1-\inv{\phi_{-}}(-s^{\star}) = \varepsilon,
	\end{equation*}
	and therefore
	\begin{equation*}
		s^{\star} = -\phi_{-}(1-\varepsilon) = \frac{2^{\worstout}-1}{\rho}\;.
	\end{equation*}
	Solving for $\worstout$ gives \eqref{eq:worst-outage}. %

	2) Next, we will prove the upper bound in an analogue way.
	Instead of an upper bound on the probability, a lower bound is used,
	\begin{equation}
		m(s) = \inf\{\Pr(S < s): X_i\sim F\}.
	\end{equation}
	
	This gives the following expressions of $s^{\star}$
	\begin{align}
		s^{\star}(\bestout) &= \sup_{s\geq 0}\left\{s: m(s) < \varepsilon\right\}\\
		&= \sup_{s\geq 0}\left\{s: \inv{\phi}(s) < \varepsilon\right\}\\
		&= \phi(\varepsilon)\\
		&= \frac{2^{\bestout}-1}{\rho}.
	\end{align}
	Solving for \bestout $\;$ gives \eqref{eq:best-outage}.
\end{proof}

Since \eqref{eq:phi-equal-m} holds with equality, the derived bounds in Theorem~\ref{thm:bounds-outage-identical} are tight, i.e., an optimal coupling exists for which the $\varepsilon$-outage capacity is achieved.

\begin{example}\label{ex:linear-pdf-1}
	For illustrating the results from Theorem~\ref{thm:bounds-outage-identical}, we give a simple example. Consider the linearly decreasing \gls{pdf} from zero to a maximal value $b>0$, i.e., $f_{X_i}(x)=\frac{2}{b}\left(1-\frac{x}{b}\right), x\in[0,b]$. The inverse \gls{cdf} is then given as $G(x)=\inv{F}(x)=b(1-\sqrt{1-x})$. The function $\phi$ that we need for $\bestout$, is calculated according to \eqref{eq:phi-decreasing}. $H_a(x)$ is given by \eqref{eq:h-decreasing} as $H_a(x)=b (1 - \sqrt{x}) + b (n-1) (1 - \sqrt{1-a-(n-1) x})$. Solving the inequality in \eqref{eq:cmin-decreasing} shows that $c=0$ always fulfills it for $n\geq 3$ and therefore $c_n(a)=0$. According to \eqref{eq:phi-decreasing}, $\phi$ now follows as $\phi(a)=n\expect{X|X>G(a)}=nb\left(2 - \frac{1}{1-a} - \frac{2}{3} \sqrt{1-a} + \frac{a}{1 - a}\right)$.
	
	We can now plug this into \eqref{eq:best-outage} to get $\bestout$. The same steps can be repeated for the negative distribution with $f_{-}(x)=\frac{2}{b}\left(1+\frac{x}{b}\right), x\in[-b,0]$.
\end{example}

In communications, we usually have to deal with decreasing densities over the non-negative real numbers, e.g., Rayleigh fading which will be discussed in Section~\ref{sec:rayleigh}.
For this type of distributions, more precise statements about the bounds in Theorem~\ref{thm:bounds-outage-identical} can be given.

\begin{prop}\label{prop:zero-outage-decreasing}
	The upper bound on the zero-outage capacity $\overline{R^{0}}$ is always positive in the case of monotone decreasing densities over the non-negative real numbers and perfect \gls{csir}.
\end{prop}
\begin{proof}
	Using \eqref{eq:best-outage}, the statement to prove
	\begin{equation*}
	\overline{R^{0}}(\rho) = \log_2\left(1 + \rho \cdot \phi(0)\right) > 0
	\end{equation*}
	is equivalent to
	\begin{equation*}
	\phi(0) > 0.
	\end{equation*}
	For decreasing densities, $\phi$ is defined according to \eqref{eq:phi-decreasing}. First, we will prove $\phi(0)>0$ for the $c_n(0)=0$.
	
	1) In the case of $c_n(0)=0$, $\phi(0)$ is given as $n\expect{X|X>G(0)}$. Since $G$ is the inverse of the \gls{cdf}, this can be simplified to the expectation $n\expect{X}$. This is greater than zero by the assumptions we had on the distribution.
	
	2) In the case of $c_n(0)>0$, $\phi(0)$ is given as
	\begin{equation*}
		H_0(c_n(0)) = (n-1)G((n-1)c_n(0)) + G(1-c_n(0)) > 0.
	\end{equation*}
	Since $G$ is the inverse of the \gls{cdf}, we observe the following properties for the assumed distributions. The first extreme case is given as $G(0)=0$, since our support are the non-negative real numbers. And second, if the support is upper-bounded, the value $G(1)$ is equal to that value. Otherwise $G(1)$ goes to infinity. In both cases, the above $H_0(c_n(0))$ is greater than zero.
\end{proof}

\begin{prop}\label{prop:zero-outage-lower}
	The lower bound on the zero-outage capacity $\underline{R^{0}}$ is always zero in the case of monotone decreasing densities over the non-negative real numbers and perfect \gls{csir}.
\end{prop}
\begin{proof}
	If $X_i\sim F$ with density $f(x)$, then $-X_i\sim F$ has the density $f_{-}(x)=f(-x)$.
	Starting with this relation, $\phi_{-}$ can be written as
	\begin{equation}
	\phi_{-}(a) = \begin{cases}
	-\left(G(1-a) + (n-1) G(0)\right) & \text{if } c^{-}_{n}(a) > 0\\
	-n\expect{X|X<G(1-a)} & \text{if } c^{-}_{n}(a) = 0
	\end{cases}\,,
	\end{equation}
	where $c^{-}_{n}(a)$ is the only value evaluated according to $f_{-}$.
	
	For the lower bound in \eqref{eq:worst-outage} from Theorem~\ref{thm:bounds-outage-identical}, $\phi_{-}$ is evaluated at $1-\varepsilon$. Observing that $G(0) = 0$, it can easily be seen that for zero-outage
	\begin{equation*}
		\phi_{-}(1) = 0.
	\end{equation*} 
	Combining this with \eqref{eq:worst-outage} completes the proof.
\end{proof}
\subsection{Different Marginals}
\label{sub:different-marginals}
We now extend the results from the previous section to the case where all coefficients are allowed to have different monotone densities, i.e., $\abs{h_i}^2\sim F_i$. Again, we are interested in the best- and worst-case bounds.

\begin{thm}[Bounds on the $\varepsilon$-outage capacity for arbitrary monotone marginals]\label{thm:bounds-outage-different}
	The $\varepsilon$-capacity $R^{\varepsilon}$ of $n$ multi-connectivity fading links with monotone marginal distributions $\abs{h_i}^2\sim F_i$ and perfect \gls{csir} can be lower bounded by
	\begin{equation}\label{eq:worst-outage-different}
		\worstout(\rho) = \log_2\left(1-\rho \cdot \Phi_{-}(1-\varepsilon)\right),
	\end{equation}
	and upper bounded by
	\begin{equation}\label{eq:best-outage-different}
		\bestout(\rho) = \log_2\left(1 + \rho \cdot \Phi(\varepsilon)\right),
	\end{equation}
	where $\rho$ is the \gls{snr} and $\Phi_{-}$ is the function $\Phi$ from \eqref{eq:def-Phi} for $-\abs{h_i}^2\sim F_i$.
\end{thm}
\begin{proof}
	The proofs are identical to the ones of Theorem~\ref{thm:bounds-outage-identical}. The only difference is that $\Phi$ is used instead of $\phi$. This is done using \eqref{eq:relation-Phi-m}, which gives looser bounds in Theorem~\ref{thm:bounds-outage-different} compared to the ones in Theorem~\ref{thm:bounds-outage-identical}.
\end{proof}

It can be clearly seen that the structure of the bounds in Theorems~\ref{thm:bounds-outage-identical} and \ref{thm:bounds-outage-different} is the same. The only difference is the use of $\Phi$ instead of $\phi$. Combining \eqref{eq:relation-Phi-m} and \eqref{eq:phi-equal-m} gives $\inv{\phi}\geq \inv\Phi$, which shows that the bounds in Theorem~\ref{thm:bounds-outage-different} are looser compared to the ones in Theorem~\ref{thm:bounds-outage-identical}. However, \eqref{eq:phi-equal-m} only holds for the homogeneous case ($F_1=F_2=\dots{}=F$). Therefore, Theorem~\ref{thm:bounds-outage-different} can be applied in the more general case of arbitrary monotone marginals.

\begin{example}\label{ex:linear-pdf-2}
	Consider the same \gls{pdf} as in Example~\ref{ex:linear-pdf-1} but now with different $b_i>0$, i.e., $f_{X_i}(x_i)=\frac{2}{b_i}\left(1-\frac{x_i}{b_i}\right), x_i\in[0,b_i]$. The function $\Phi$, as defined in \eqref{eq:def-Phi}, needs the conditional expectation which is calculated for our example as $\expect{X_i|X_i>G_i(a)}=b_i\left(1 - \frac{2}{3} \sqrt{1-a}\right)$. The upper bound is now evaluated according to \eqref{eq:best-outage-different} as
	\begin{equation*}
	\bestout(\rho) = \log_2\left(1+\rho\left(1-\frac{2}{3} \sqrt{1-\varepsilon}\right)\sum_{i=1}^{n} b_i\right)\,.
	\end{equation*}
	Again, the calculations are analogue of $\Phi_{-}$ and $\worstout$. You can find an interactive step-by-step version at \cite{BesserGitlab}.
\end{example}

\begin{rem}\label{rem:bounds-overlap}
	For $c_n(a)=0$, $\phi$ defined in \eqref{eq:phi-decreasing} and \eqref{eq:phi-increasing} reduces to $n\expect{X|X>G(a)}$. This is obviously very similar to the definition of $\Phi$ in \eqref{eq:def-Phi}. In fact, they are the same in the homogeneous case. In this case, the bounds from Theorem~\ref{thm:bounds-outage-identical} and \ref{thm:bounds-outage-different} are identical. An example where this happens will be given in the next section.
\end{rem}

Proposition~\ref{prop:zero-outage-decreasing} can also be extended to arbitrary monotone decreasing marginals on the non-negative real numbers as follows.
\begin{prop}\label{prop:upper-bound-positive-different}
	The upper bound on the zero-outage capacity $\overline{R^{0}}$ is always positive in the case of $n$ arbitrary monotone decreasing densities over the non-negative real numbers and perfect \gls{csir}.
\end{prop}
\begin{proof}
	Similar to Theorems~\ref{thm:bounds-outage-identical} and \ref{thm:bounds-outage-different}, the proof is analogue to the one of Proposition~\ref{prop:zero-outage-decreasing}.
	However, in this case, we only need to prove that $\Phi(0) > 0$.
	
	Following the definition of $\Phi$ from \eqref{eq:def-Phi} and the argumentation in the proof of Proposition~\ref{prop:zero-outage-decreasing} yields
	\begin{equation*}
	\expect[X_i\sim F_i]{X_i|X_i > G_i(0)} = \expect[X_i\sim F_i]{X_i} > 0,
	\end{equation*}
	where the last inequality follows from the assumptions on the marginal distributions. 
\end{proof}

Please note that Proposition~\ref{prop:upper-bound-positive-different} only states that the upper bound on the zero-outage capacity is positive. Since the bound might be loose due to \eqref{eq:relation-Phi-m}, this does not necessarily mean that the actual zero-outage capacity is positive in the best case.
\section{Zero-Outage Capacity for Perfect Channel State Information}\label{sec:perfect-csit}

A special case of the $\varepsilon$-outage capacity, which is of particular interest, is the zero-outage capacity. In this section, we will take a closer look at it and show how additional perfect \gls{csi} at the transmitter can improve the bounds from the previous section, where only perfect \gls{csi} at the receiver was assumed.
Throughout this section, we will assume that the fading coefficients $\abs{h_i}^2$ have the same monotone decreasing density over the non-negative real numbers.
An example for Rayleigh fading will be given in Section~\ref{sub:perfect-csi-rayleigh}.

As shown in \cite{Jorswieck2007}, the zero outage capacity $R^{0}$ with perfect \gls{csi} at the transmitter can be written as
\begin{equation}\label{eq:zero-out-perfect-csi}
R^{0}_{\text{CSIT}} = \log_2\left(1+\rho \left(\expect{\dfrac{1}{\sum_{i=1}^{n}\abs{h_i}^2}}\right)^{-1} \right)\,.
\end{equation}

\subsection{Best Case}
The best case is obtained when $\expect{\frac{1}{\sum_{i=1}^{n}\abs{h_i}^2}}$ is minimized. This problem is solved in \cite{Wang2011}, where the authors show that the minimum of the expected value of a convex function of the sum of random variables with a monotone density is attained for a particular dependence structure, see (\ref{eq:minimization-problem-wang}).  

\begin{lem}\label{lem:zero-out-perfect-csi-best}
	The best-case zero-outage capacity for fading channels with monotone decreasing densities and perfect \gls{csi} at the transmitter is given by
	\begin{equation}\label{eq:zero-out-perfect-csi-best}
	\overline{R^{0}_{\text{CSIT}}} = \log_2\left(1+\frac{\rho\cdot{}H(c_n)}{H(c_n) n \int_{1-c_n}^{1}\frac{1}{H(1-x)}\diff{x} + 1-nc_n}\right)\,,
	\end{equation}
	where $H$ and $c_n$ are given by $H_0$ from \eqref{eq:h-decreasing} and $c_n(0)$ from \eqref{eq:cmin-decreasing}, respectively.
\end{lem}
\begin{proof}
	Using the minimization problem in \eqref{eq:minimization-problem-wang} and \cite[Proof of Prop.~3.1]{Wang2013}, the relation in \eqref{eq:perfect-csi-integral-expect-g} at the bottom of the page can be derived, where $g$ is a convex function.%
	\begin{figure*}[b]
		\noindent\rule{\textwidth}{.5pt}
		\begin{equation}\label{eq:perfect-csi-integral-expect-g}
		\min_{\abs{h_i}^2\sim F} \expect{g\left(\sum_{i=1}^{n}\abs{h_i}^2\right)} = n\int_{1-c_n(0)}^{1}g\left(H_0(1-x)\right)\diff{x} + g\left(H_0(c_n(0))\right)\left(1-nc_n(0)\right)
		\end{equation}
	\end{figure*}
	The detailed derivation can be found in Appendix~\ref{app:minimization-problem}.
	According to \eqref{eq:zero-out-perfect-csi}, we set $g(x)=\frac{1}{x}$ for our considered problem, and get \eqref{eq:zero-out-perfect-csi-best}.
\end{proof}

Based on this lemma, we are able to make the following statement about the best-case zero-outage capacities in the scenarios with and without perfect \gls{csi} at the transmitter.
\begin{thm}\label{thm:gap-best-case-zero-out-csit}
	In the case of fading coefficients $\abs{h_i}^2$ with monotone decreasing densities, perfect \gls{csi} at the transmitter can increase the highest achievable zero-outage capacity
	\begin{equation}\label{eq:zero-out-best-perfect-csi-improvement}
	\overline{R^{0}_{\text{CSIT}}} \geq \overline{R^{0}}
	\end{equation}
\end{thm}
\begin{proof}
	For $c_n=0$, the proof is trivial.
	
	For $c_n>0$, $\overline{R^{0}_{\text{CSIT}}}$ in \eqref{eq:zero-out-perfect-csi-best} can be easily compared to the best-case zero-outage capacity $\overline{R^{0}}$ in the case of no \gls{csi} at the transmitter from Theorem~\ref{thm:bounds-outage-identical}
	\begin{equation*}
	\overline{R^{0}} = \log_2\left(1 + \rho\cdot{}H_0(c_n(0))\right)\,.
	\end{equation*}
	We observe that \eqref{eq:zero-out-best-perfect-csi-improvement} holds, if
	\begin{equation*}
		1\geq H(c_n) n \int_{1-c_n}^{1}\frac{1}{H(1-x)}\diff{x} + 1-n c_n\,.
	\end{equation*}
	Using some minor manipulations as well as the mean value theorem, this can be restated as
	\begin{equation}
	n c_n \geq n c_n \frac{H(c_n)}{H(1-\xi)},
	\end{equation}
	where $\xi\in[1-c_n, 1]$ is the mean value (and therefore $1-\xi\in[0, c_n]$). As stated in \cite[Lemma~3.3]{Wang2013}, $H(x)$ is decreasing for $x\in[0, c_n]$ and attains its minimum at $c_n$. Therefore, $H(c_n)\leq H(1-\xi)$ holds, which completes the proof.
\end{proof}

\subsection{Worst Case}
We need the maximum of $\expect{\frac{1}{X_1 + \cdots{} + X_n}}$ for minimizing the zero-outage capacity with perfect \gls{csi} at the transmitter. The following corollary provides the corresponding joint distribution which achieves this maximum. 

\begin{cor}\label{cor:maximum-expect-inverse-sum}
	The maximum of $\expect{\frac{1}{X_1 + \cdots{} + X_n}}$ is achieved for comonotonic $(X_1,...,X_n)$, i.e., for $F(X_1,...,X_n) = \min \{ F_1(X_1),..., F_n(X_n) \}$, or alternatively if $X_i \sim F^{-1}_i(U)$ for $i=1,...,n$ and uniform $U \sim \mathcal{U}[0,1]$.
\end{cor}

\begin{proof}
The corollary follows from \cite[Theorem 2.1]{puccetti2015}. It is also mentioned in \cite[Section 4.1]{Wang2016}.  
\end{proof}

The comonotonicity of $(X_1,...,X_n)$ means that all variables are completely dependent or completely coupled \cite{Thorisson1998}, i.e., $U_1=\dots{}=U_n=U$ with probability one and $X_i=G_i(U)$. Instead of diversity, this provides only the $n$-fold sum of the signal.
\section{Rayleigh Fading Channels}\label{sec:rayleigh}
In this section, the special case of Rayleigh fading is considered. This channel model is commonly used for modeling mobile channels~\cite{Biglieri1998}.

Besides Rayleigh fading, other fading distributions with monotone marginals are known, e.g., Weibull-fading \cite{Yacoub2002} and $\alpha$-$\mu$-fading \cite{Alexandropoulos2007}. In this section, we only show the example of Rayleigh fading in detail, however, all calculations can be adapted for other fading distributions.

In this case, the amplitude of the fading coefficients $\abs{h_i}$ follows a Rayleigh distribution.
Recall that the $\varepsilon$-outage capacity can be determined as
\begin{equation}
	R^{\varepsilon} = \sup_{R\geq 0}\left\{R \in \mathbb{R}: \Pr\left(\sum_{i=1}^{n} \abs{h_i}^2 <\frac{2^R-1}{\rho}\right)<\varepsilon\right\}.
\end{equation}
Since $\abs{h_i}$ is Rayleigh distributed, $\abs{h_i}^2$ is exponentially distributed with $\abs{h_i}^2\sim \exp(\lambda_i)$~\cite[Ch.~39]{Forbes2010}.

The exponential distribution has a monotone decreasing density $f(x)=\lambda\exp(-\lambda x)$ which allows us to apply Theorems~\ref{thm:bounds-outage-identical} and \ref{thm:bounds-outage-different} to determine lower and upper bounds for the outage capacity. 

The expressions for the different functions necessary to compute the following examples are listed in Appendix~\ref{app:rayleigh-fading-example}.
The details, including Python code, for reproducing all of the calculations in this section can be found at~\cite{BesserGitlab}.

\subsection{Identical Distributions}
At first, the case of perfect \gls{csir} and where all $\lambda_i=1$ are the same is considered. In this case, Theorem~\ref{thm:bounds-outage-identical} can be applied. In the following, the resulting bounds on the $\varepsilon$-outage capacity are shown along with the \gls{iid} case.

\subsubsection{Upper Bound}
The upper bound $\bestout$ on the $\varepsilon$-outage capacity is calculated according to \eqref{eq:best-outage}. In order to do this, $\phi$ needs to be calculated according to \eqref{eq:phi-decreasing} since $\abs{h_i}^2\sim\exp(1)$ has a decreasing density. Hence, $c_n(a)$ needs to be determined according to \eqref{eq:cmin-decreasing}.
In the case of exponentially distributed variables, this yields the following inequality
\begin{equation}\label{eq:cmin-rayleigh-decreasing}
	 (a-1)\log\left(\frac{1-a-(n-1)c}{c}\right) \geq n(a+cn-1).
\end{equation}
However, this cannot be solved by a closed-form expression but rather numerically. Figure~\ref{fig:cmin-exp} shows the calculated values of $c_n(a)$ for different $n$ over $a$.
\begin{figure}
	\centering
	\begin{tikzpicture}%
\begin{axis}[
	xlabel={$a$},
	ylabel={$c_{n}(a)$},
	width=\linewidth,
	height=.3\textheight,
	legend entries={{$n=3$}, {$n=5$}, {$n=10$}},
	legend style={at={(0.98,0.45)},anchor=north east},
	ymode=log,
	ymajorgrids,
	xmajorgrids,
	xminorgrids,
	grid style={line width=.1pt, draw=gray!20},
	major grid style={line width=.25pt,draw=gray!40},
]
\addplot[thick, plot0] table[x=a,y=cmin] {data/cmin-exp-n3.dat};
\addplot[thick, plot1, dashed] table[x=a,y=cmin] {data/cmin-exp-n5.dat};
\addplot[thick, plot2, dashdotted] table[x=a,y=cmin] {data/cmin-exp-n10.dat};
\end{axis}
\end{tikzpicture}%
	\caption{Computed values of $c_n(a)$ according to \eqref{eq:cmin-decreasing} over $a$ for different $n$ with $\abs{h_i}^2\sim\exp(1), i=1, \dots{}, n$.}
	\label{fig:cmin-exp}
\end{figure}
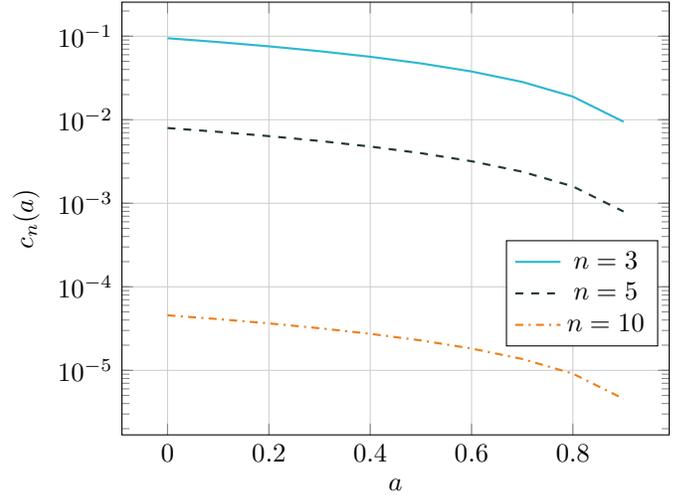
After determining $c_n(a)$, they are used to calculate $\phi$ and $\bestout$. For our considered case, $\phi$ can be expressed as
\begin{equation}\label{eq:phi-rayleigh-identical}
\phi(a) = \begin{cases}
	H_{a}(c_{n}(a)) & \text{if } c_{n}(a) > 0\\
	n(1-\log(1-a)) & \text{if } c_{n}(a) = 0
\end{cases}\,.
\end{equation}
For $a\in[0, 1)$, it can be seen from \eqref{eq:cmin-rayleigh-decreasing} that $c_n(a)>0$. Therefore, $\phi$ can be simplified to
\begin{align}
	\phi(a) &= H_{a}(c_{n}(a))\notag\\ 
			&= -(n-1)\log(1-a-(n-1)c_{n}(a)) - \log(c_{n}(a)).\label{eq:phi-rayleigh-decreasing}
\end{align}

An example for Rayleigh fading with $n=5$ is presented in Fig.~\ref{fig:outage-rayleigh} and \ref{fig:outage-rayleigh-low-eps} later.
As described in Proposition~\ref{prop:zero-outage-decreasing}, the capacity in the best case is always positive and even the zero-outage capacity has a value of around $4.06$ bits per channel use. It increases with increasing $\varepsilon$ and tends towards infinity for $\varepsilon\to 1$. Please note that the usual operating point is $\varepsilon<10^{-2}$ or smaller for \gls{urllc}.

\begin{cor}[Best-case Zero-Outage Capacity of Rayleigh Fading Links]
	For $n$ Rayleigh fading links with perfect \gls{csir}, the best-case zero-outage capacity $\overline{R^{0}}$ is upper-bounded by
	\begin{equation}\label{eq:zero-outage-rayleigh-upper-bound-best}
	\overline{R^{0}}(\rho) \leq \log_2\left(1+\rho\cdot{}n\log(n)\right).
	\end{equation}
\end{cor}
\begin{proof}
	As stated in \cite[Lemma~3.3]{Wang2013}, $\phi(a)$ is strictly increasing. By definition (cf.~\eqref{eq:cmin-decreasing}), $c_n(a)\leq\frac{1}{n}$ holds.
	Combining this with \eqref{eq:phi-rayleigh-decreasing} gives,
	\begin{equation}
		\phi(0) \leq n\log(n).
	\end{equation}
	Applying this to \eqref{eq:best-outage} yields \eqref{eq:zero-outage-rayleigh-upper-bound-best}.
\end{proof}

\begin{example}[Two User Rayleigh Fading]
	As a short example of the above corollary, we consider the two-user scenario with Rayleigh fading links. %
	In this example, the \gls{cdf} and its inverse are defined as $F(x)=1-\exp(-x)$ and $G(x)=-\log(1-x)$, respectively.
	Solving \eqref{eq:cmin-decreasing} for $n=2$ and $a=0$ gives the solution $c_2(0)=\frac{1}{2}$. This yields
	\begin{align}
	\phi(0) &= H_0\left(\frac{1}{2}\right)\\
	&= 2\log\left(2\right),
	\end{align}
	and
	\begin{equation}
	\overline{R^{0}}(\rho) = \log_2\left(1+2\rho\log(2)\right).
	\end{equation}
	The same result of this special case is also derived in \cite{Jorswieck2019} by direct methods and a special construction of joint distribution for two random variables. Note that in this case of $n=2$, the upper bound formulated in the previous corollary is tight.
\end{example}

\subsubsection{Lower Bound}
The lower bound $\worstout$ on the $\varepsilon$-outage capacity is calculated according to \eqref{eq:worst-outage}. The approach is very similar to the one of calculating the upper bound in the previous section.
However, $\phi_{-}$ is required and all necessary steps have to be calculated for the random variables $-\abs{h_i}^2\sim \exp(1)$. The distribution of these variables has the increasing density $f_{-}(x)=\exp(x)$ on the support $(-\infty, 0]$.
Using this gives the following inequality to be solved for $c$ according to \eqref{eq:cmin-increasing}
\begin{equation}
	\frac{n(a+nc-1)}{1+(n-1)a} \leq \log(a+c) - \log(1-(n-1)c)\,.
\end{equation}

Depending on $c_{n}(a)$, $\phi_{-}$ can be expressed as
\begin{equation}\label{eq:phi-minus-rayleigh-identical}
\phi_{-}(a) = \begin{cases}
\log(a) & \text{if } c_{n}(a) > 0\\
n\cdot{}\frac{a - a\log(a) - 1}{1-a} & \text{if } c_{n}(a) = 0
\end{cases}.
\end{equation}

\begin{rem}\label{rem:min-a-for-c-zero}
	It should be noted that in this case the solution $c_{n}(a)=0$ occurs, if the following holds
	\begin{equation}\label{eq:condition-a-c-zero}
		\frac{n(a-1)}{1+(n-1)a} \leq \log(a).
	\end{equation}
	This inequality can be solved numerically and gives a minimum $a$, which is decreasing with $n$. As an example, the minimum $a$ for which the inequality holds is around $0.117$ and $4.56\cdot{}10^{-5}$ for $n=3$ and $n=10$, respectively.
	For the lower bound on the outage-capacity in Theorem~\ref{thm:bounds-outage-identical}, $\phi_{-}$ is evaluated at the argument $1-\varepsilon$. An interesting application of the outage-capacity is \gls{urllc}, where $\varepsilon$ close to zero are considered~\cite{Bennis2018}. In this case, $1-\varepsilon$ is close to one and therefore the above inequality holds. Therefore, $c_{n}(1-\varepsilon)$ is zero in this case.
\end{rem}

The calculated $\worstout$ for a Rayleigh fading channel with $n=5$ and $\rho=\SI{5}{\decibel}$ is shown in Fig.~\ref{fig:outage-rayleigh} and \ref{fig:outage-rayleigh-low-eps}.
One noticeable point is that the worst-case zero-outage capacity is zero. However, for positive $\varepsilon$, it attains positive values and tends towards infinity for $\varepsilon\to 1$. This is expected since a value of $\varepsilon=1$ implies that an arbitrarily large number of transmission errors is tolerated.

\subsubsection{Comonotonic Coefficients}
One special case are comonotonic coefficients, i.e., $\abs{h_i}^2\sim G(U)$ with $U\sim\mathcal{U}[0, 1]$. In this case, the following holds
\begin{align*}
\Pr\left(\sum_{i=1}^{n}\abs{h_i}^2 \leq \frac{2^R - 1}{\rho}\right) &= \Pr\left(\sum_{i=1}^{n}\abs{h_1}^2 \leq \frac{2^R - 1}{\rho}\right)\\
&= \Pr\left(\abs{h_1}^2 \leq \frac{2^R - 1}{n\rho}\right)\\
&= F\left(\frac{2^R - 1}{n\rho}\right)\,,
\end{align*}
which combined with \eqref{eq:def-outage} yields the following expression for the $\varepsilon$-outage capacity
\begin{equation}\label{eq:outage-comonotonic-rayleigh}
R^{\varepsilon}_{\text{comon}} = \log_2\left(1 - \rho\cdot{}n\log(1-\varepsilon)\right)\,.
\end{equation}
The outage capacity for this case is also shown for comparison in Fig.~\ref{fig:outage-rayleigh} and \ref{fig:outage-rayleigh-low-eps}.

\subsubsection{Independent Coefficients}
In the case of \gls{iid} fading coefficients $h_i$, the outage capacity can be given as a closed-form expression.
Since the sum of independent exponentially distributed random variables with the same mean is Gamma-distributed~\cite[Ch. 17.2]{Forbes2010}, the $\varepsilon$-outage capacity in the independent case is given as
\begin{equation}
	R^{\varepsilon}_{\text{iid}} = \log_2\left(1 + \rho\inv{P}\left(n, \varepsilon\right)\right),
\end{equation}
where $P(a, x)$ is the regularized incomplete Gamma-function~\cite[Eq. 6.5.1]{Abramowitz1972} (which is the \gls{cdf} of the Gamma-distribution \cite{Forbes2010}).

\begin{figure}
	\centering
	\begin{tikzpicture}%
\begin{axis}[
	xlabel={$\varepsilon$},
	ylabel={$R^{\varepsilon}$},
	width=\linewidth,
	height=.3\textheight,
	legend entries={{Best Case \bestout}, {Worst Case \worstout}, {Iid Case}, {Comon. Case}},
	legend pos=south east,
	legend columns=2,
	legend cell align=left,
	ymajorgrids,
	xmajorgrids,
	xminorgrids,
	grid style={line width=.1pt, draw=gray!20},
	major grid style={line width=.25pt,draw=gray!40},
]
\addplot[thick, plot0] table[x=eps,y=capacity] {data/outage-rayleigh-Best_Case-n5-snr3.16.dat};
\addplot[thick, plot1, dashed] table[x=eps,y=capacity] {data/outage-rayleigh-Worst_Case-n5-snr3.16.dat};
\addplot[thick, plot2, densely dashdotted] table[x=eps,y=capacity] {data/outage-rayleigh-IID_Case-n5-snr3.16.dat};
\addplot[thick, plot3, densely dotted] table[x=eps,y=capacity] {data/outage-rayleigh-Comon_Case-n5-snr3.16.dat};
\end{axis}
\end{tikzpicture}%
	\caption{Lower and upper bound on the $\varepsilon$-outage capacity of a Rayleigh fading channel with $\abs{h_i}^2\sim\exp(1), i=1, \dots{}, 5$ at $\rho=\SI{5}{\decibel}$. In addition, the special cases of independent and comonotonic $\abs{h_i}^2$ are shown.}
	\label{fig:outage-rayleigh}
\end{figure}

\begin{figure}
	\centering
	\begin{tikzpicture}%
\begin{axis}[
	xlabel={$\varepsilon$},
	ylabel={$R^{\varepsilon}$},
	width=.95\linewidth,
	height=.3\textheight,
	xmin=1e-7,
	xmax=1e-1,
	ymin=1e-6,
	ymax=10,
	xmode=log,
	ymode=log,
	legend entries={{Best Case \bestout}, {Worst Case \worstout}, {Iid Case}, {Comon. Case}},
	legend pos=south east,
	legend columns=2,
	legend cell align=left,
	ymajorgrids,
	xmajorgrids,
	xminorgrids,
	grid style={line width=.1pt, draw=gray!20},
	major grid style={line width=.25pt,draw=gray!40},
]
\addplot[thick, plot0] table[x=eps,y=best] {data/rayleigh-outage-capacities-small-eps-n5-snr5dB.dat};
\addplot[thick, plot1, dashed] table[x=eps,y=worst] {data/rayleigh-outage-capacities-small-eps-n5-snr5dB.dat};
\addplot[thick, plot2, densely dashdotted] table[x=eps,y=iid] {data/rayleigh-outage-capacities-small-eps-n5-snr5dB.dat};
\addplot[thick, plot3, densely dotted] table[x=eps,y=comon] {data/rayleigh-outage-capacities-small-eps-n5-snr5dB.dat};
\end{axis}
\end{tikzpicture}%
	\caption{Lower and upper bound on the $\varepsilon$-outage capacity of a Rayleigh fading channel with $\abs{h_i}^2\sim\exp(1), i=1, \dots{}, 5$ at $\rho=\SI{5}{\decibel}$. In addition, the special cases of independent and comonotonic $\abs{h_i}^2$ are shown.}
	\label{fig:outage-rayleigh-low-eps}
\end{figure}

\subsubsection{Summary}
Figures~\ref{fig:outage-rayleigh} and \ref{fig:outage-rayleigh-low-eps} show the $\varepsilon$-outage capacity for Rayleigh fading coefficients with $n=5$ and an \gls{snr} of $\rho=\SI{5}{\decibel}$ in the best case and worst case, along with the cases of \gls{iid} and comonotonic coefficients.

As expected, the \gls{iid} and comonotonic cases lie in between the other two and show a similar behavior. All curves increase with increasing $\varepsilon$ and tends towards infinity for $\varepsilon\to 1$. For small $\varepsilon$, the outage capacity for comonotic coefficients is smaller than for independent ones. However, there is a $\varepsilon$ for which the curves meet and the capacity of the \gls{iid} case becomes smaller. For the interesting operating region shown in Fig.~\ref{fig:outage-rayleigh-low-eps}, the \gls{iid} case performs better than the comonotonic case with a significantly higher $\varepsilon$-outage capacity.

Interestingly, all zero-outage capacities are zero except in the best case. This is due to the considered model of no \gls{csi} at the transmitter. In the case of perfect \gls{csi} at the transmitter, the zero-outage capacity in the \gls{iid} case increases to a positive value, which will be shown in Section~\ref{sub:perfect-csi-rayleigh}.

Since we are often interested in $\varepsilon$ close to zero, e.g., in the context of \gls{urllc}, a interesting property of the different curves in Fig.~\ref{fig:outage-rayleigh} is the slope at $\varepsilon=0$. As can be seen from the figure, the slope in the worst case is the lowest with $\frac{n\rho}{2\log 2}$. For comonotonic coefficients, the slope increases to $\frac{n\rho}{\log 2}$. In the case of \gls{iid} coefficients, the slope at $\varepsilon=0$ approaches infinity. This means that allowing even a very small outage probability increases the $\varepsilon$-outage capacity drastically in the case of independent coefficients.

\subsection{Different Marginals}
Next, a more general scenario of Rayleigh fading channels is considered. In this setting, we now consider different fading variances $\inv{\lambda_i}$ for the users. For a fair comparison, we constrain the sum of the variance to be constant and equal to the number of users, i.e., $\sum_{i=1}^{n}\inv{\lambda_i}=n$ with $\abs{h_i}^2\sim\exp(\lambda_i)$.
This could happen when one user moves closer toward the base station while another one moves away~\cite{Jorswieck2006}. 

\subsubsection{Upper Bound}
The upper bound on the $\varepsilon$-outage capacity for different marginal distributions is calculated according to \eqref{eq:best-outage-different}. Therefore, we need to calculate $\Phi(\varepsilon)$ first. The function $\Phi$ is given in \eqref{eq:def-Phi} and can be determined as follows for the exponential distribution
\begin{align}\label{eq:Phi-rayleigh-different}
\Phi(\varepsilon) &= \sum_{i=1}^{n}\expect[X_i\sim\exp(\lambda_i)]{X_i|X_i\geq G_i(\varepsilon)}\\
&= \sum_{i=1}^{n}\expect{X_i} + G_i(\varepsilon)\\
&= (1-\log(1-\varepsilon)) \sum_{i=1}^{n}\frac{1}{\lambda_i}\\
&= (1-\log(1-\varepsilon))\cdot{}n\label{eq:Phi-rayleigh-different-last} %
\end{align}
where the \enquote{lack of memory property} of the exponential distribution~\cite{Everitt2010} and our assumption of the constant sum of the expected values are used.

Finally, the upper bound $\bestout$ can be formulated according to \eqref{eq:best-outage-different} as
\begin{equation}
	\bestout(\rho) = \log_2\left(1 + \rho \cdot{} n \cdot{} (1-\log(1-\varepsilon))\right)\,.
\end{equation}

\subsubsection{Lower Bound}
The lower bound $\worstout$ can be determined in an analogue way according to \eqref{eq:worst-outage-different}. In order to do this, the function $\Phi_{-}$ has to be calculated.
This can be done as follows
\begin{align}\label{eq:Phi-minus-rayleigh-different}
\Phi_{-}(a) &= \sum_{i=1}^{n}\expect[-X_i\sim\exp(\lambda_i)]{X_i|X_i\geq G_i(a)}\\
&= \sum_{i=1}^{n}\int_{-\infty}^{\infty}{x_i f_{X_i|X_i\geq G_i(a)}(x_i)}\diff{x_i}\\
&= \sum_{i=1}^{n}\frac{1}{\Pr(X_i\geq G_i(a))}\int_{G_i(a)}^{\infty}x_i f_{X_i}(x_i) \diff{x_i}\\
&= \sum_{i=1}^{n}\frac{1}{1-\underbrace{F_i(G_i(a))}_{a}}\int_{G_i(a)}^{0}x_i \lambda_i\exp(\lambda_i x_i) \diff{x_i}\\
&= \frac{a - a\log(a) - 1}{1-a} \sum_{i=1}^{n}\frac{1}{\lambda_i}\\
&= \frac{a - a\log(a) - 1}{1-a} \cdot{} n\label{eq:Phi-minus-rayleigh-different-last-eq} %
\end{align}

Therefore, the lower bound on the $\varepsilon$-outage capacity can be formulated as
\begin{equation}
	\worstout(\rho) = \log_2\left(1 + \rho \cdot{} n \cdot{} \frac{\varepsilon + (1-\varepsilon)\log(1-\varepsilon)}{\varepsilon}\right)\,.
\end{equation}

\begin{rem}
	As mentioned in Remark~\ref{rem:bounds-overlap}, it is possible that the in general looser bounds from Theorem~\ref{thm:bounds-outage-different} for different marginals are identical to the ones from Theorem~\ref{thm:bounds-outage-identical} for identical marginals.
	
	The given example of Rayleigh fading shows this behavior for the lower bound $\worstout$. In the homogeneous case of $\abs{h_i}^2\sim\exp(1)$, $\phi_{-}(a)$ is given by \eqref{eq:phi-minus-rayleigh-identical}.
	For $c_n(a)=0$, this is identical to $\Phi_{-}(a)$ for Rayleigh fading with different marginals $\abs{h_i}^2\sim\exp(\lambda_i)$ and $\sum_{i=1}^{n}\inv{\lambda_i}=n$ shown in \eqref{eq:Phi-minus-rayleigh-different-last-eq}.
	
	This means that, in the case of $c_n(a)=0$, the lower bounds $\worstout$ in the cases of identical and different marginals are the same.
	Equation~\eqref{eq:condition-a-c-zero} states the condition that the case $c_n(a)=0$ occurs. As stated in Remark~\ref{rem:min-a-for-c-zero}, there is a minimum $a$ (decreasing with $n$) for which the condition holds.
	
	Since we are usually interested in the ultra-reliable scenarios, our considered $\varepsilon$ are close to zero. Hence, $a = 1-\varepsilon$ is close to one and greater than the minimum $a$ such that condition \eqref{eq:condition-a-c-zero} holds. Therefore, $c_n(1-\varepsilon)$ is zero and $\phi_{-}(1-\varepsilon)=\Phi_{-}(1-\varepsilon)$.
\end{rem}

\subsection{Perfect CSI at the Transmitter}\label{sub:perfect-csi-rayleigh}
We now go back to the assumption from the previous section that $\abs{h_i}^2\sim \exp(\lambda)$, but now also assume that the transmitter has perfect \gls{csi}.
In the following, we will only focus on the zero-outage capacity and compare it to the scenario without \gls{csit}.

For convenience, we will restate the zero-outage capacity in the case of perfect \gls{csit} given in \eqref{eq:zero-out-perfect-csi}
\begin{equation*}
R^{0}_{\text{CSIT}} = \log_2\left(1+\rho \left( \expect{\dfrac{1}{\sum_{i=1}^{n}\abs{h_i}^2}}\right)^{-1}\right)\,.
\end{equation*}

\subsubsection{Upper Bound}
The upper bound on the zero-outage capacity with perfect \gls{csit} for Rayleigh fading can be evaluated according to \eqref{eq:zero-out-perfect-csi-best} in Lemma~\ref{lem:zero-out-perfect-csi-best}. The expression for $H$ in this case is given in Table~\ref{tab:formula-rayleigh} in the appendix.
Figure~\ref{fig:outage-rayleigh-perfect-csit} shows this upper bound on the zero-outage capacity for an \gls{snr} of \SI{5}{\decibel} for different $n$. Since there exists no closed-form solution for determining $c_n$, the values are evaluated numerically.

\subsubsection{Lower Bound}
As stated in Corollary~\ref{cor:maximum-expect-inverse-sum}, the lower bound is attained for comonotonic variables $(\abs{h_1}^2, \dots{}, \abs{h_n}^2)$. In our case, $\abs{h_i}\sim \exp(\lambda)=G(U)$ with $U\sim\mathcal{U}[0, 1]$ which gives the following solution to the maximum value of the expected value
\begin{align}
	\max \expect[\abs{h_i}^2\sim\exp(\lambda)]{\frac{1}{\sum_{i=1}^{n}\abs{h_i}^2}} &= \expect[{U\sim\mathcal{U}[0, 1]}]{\frac{1}{\sum_{i=1}^{n}G(U)}}\\
	&= \frac{1}{n}\int_{0}^{1}\frac{1}{G(u)}\diff{u}\\
	&= -\frac{1}{n}\int_{0}^{1}\frac{\lambda}{\log(1-u)}\diff{u}\,.
\end{align}
Since this grows to infinity, the zero-outage capacity from \eqref{eq:zero-out-perfect-csi} goes down to zero in this worst-case scenario of comonotonic fading coefficients, i.e., $\underline{R^{0}_{\text{CSIT}}} = 0$.

\subsubsection{Independent Coefficients}
As stated previously, the sum of independent exponentially distributed random variables $X_i\sim\exp(\lambda)$ follows a Gamma-distribution $\sum_{i=1}^{n}X_i\sim \Gamma(n, \inv{\lambda})$ (shape and scale)~\cite[Ch. 17.2]{Forbes2010}.
In our case, we need the expected value of the inverse of such a Gamma-distributed variable. This is distributed according to an inverse-gamma distribution $\text{IG}(n, \inv{\lambda})$ which has the expected value of~\cite[Ch.~22.4]{Forbes2010}
\begin{equation*}
\expect[S\sim\Gamma(n, \inv{\lambda})]{\frac{1}{S}} = \frac{\lambda}{n-1}\,.
\end{equation*}

Combining this expression with \eqref{eq:zero-out-perfect-csi} gives the zero-outage capacity with perfect \gls{csit} for independent Rayleigh fading coefficients $\abs{h_i}\sim\exp(\lambda)$ as
\begin{equation}
R^{0}_{\text{CSIT, iid}} = \log_2\left(1+\frac{\rho(n-1)}{\lambda}\right)\,.
\end{equation}

\subsubsection{Summary}
Figure~\ref{fig:outage-rayleigh-perfect-csit} shows the best-case and \gls{iid} case zero-outage capacities with perfect \gls{csit} for varying $n$ and $\abs{h_i}^2\sim\exp(1)$. For comparison, the best-case curve from the scenario without \gls{csit} is also given. Note that the zero-outage for \gls{iid} coefficients is always zero in this case.
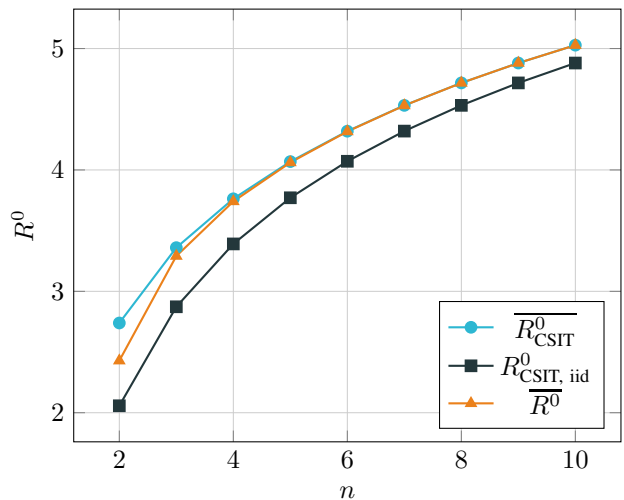
\begin{figure}[t]%
	\centering
	\begin{tikzpicture}%
\begin{axis}[
	xlabel={$n$},
	ylabel={$R^{0}$},
	width=\linewidth,
	height=.3\textheight,
	legend entries={{$\overline{R^{0}_{\text{CSIT}}}$}, {${R^{0}_{\text{CSIT, iid}}}$}, {$\overline{R^{0}}$}},
	legend pos=south east,
	ymajorgrids,
	xmajorgrids,
	xminorgrids,
	grid style={line width=.1pt, draw=gray!20},
	major grid style={line width=.25pt,draw=gray!40},
]
\addplot[thick, plot0, mark=*] table[x=n,y=capacity] {data/perfect-csi-best.dat};
\addplot[thick, plot1, mark=square*] table[x=n,y=capacity] {data/perfect-csi-iid.dat};

\addplot[thick, plot2, mark=triangle*] table[x=n,y=capacity] {data/zero_out-no-csit-best.dat};
\end{axis}
\end{tikzpicture}%
	\caption{Comparison of the best-case and iid case zero-outage capacity with and without perfect CSI at the transmitter for $\abs{h_i}^2\sim\exp(1)$ and $\rho=\SI{5}{\decibel}$.}
	\label{fig:outage-rayleigh-perfect-csit}
\end{figure}

The first obvious thing which can be seen is that there is a gap between the \gls{iid} and the best-case zero-outage capacity in the case of perfect \gls{csit}. More interestingly, the gap between the best-case with and without perfect \gls{csit} vanishes as $n$ grows. This behavior is expected according to Theorem~\ref{thm:gap-best-case-zero-out-csit} and the fact that $c_n$ approaches zero for increasing $n$. In the case of $c_n=0$, the gap would be equal to zero.
In conclusion, this shows that the benefit of having perfect \gls{csit} vanishes with increasing $n$ in the best-case scenario.

However, it should be emphasized that this is only true for the best case scenario. Without perfect \gls{csit}, the zero-outage capacity drops to zero for many other dependence structures of the fading coefficients, e.g., the independent case. However, with perfect \gls{csit} the zero-outage capacity is positive in the \gls{iid} case.

\section{Conclusion}\label{sec:conclusion}

In this work, lower and upper bounds on the $\varepsilon$-outage capacity including the zero-outage capacity of fading channels with monotone marginal densities are derived. Since the individual fading coefficients were not assumed to be independent, these bounds hold for the general case of arbitrary dependency between the fading channels.

A remarkable result is that the zero-outage capacity of $n\geq 2$ Rayleigh fading links without instantaneous \gls{csit} can be greater than zero. However, a particular dependence structure between the fading coefficients is required. In the worst-case (and even the \gls{iid}) case, the zero-outage capacity is equal to zero.
Note that one challenge follows from the fact that \cite{Wang2013} only proves the existence of such a dependence structure. Deriving an explicit form remains an open problem, not to speak about the \enquote*{practical construction} of the corresponding propagation scenario. 

This leads to the question about the scenario in which the channel gains are negatively dependent. As described in the introduction, \cite{Haber1974,Akki1985} have shown how to create negatively correlated diversity branches in the frequency domain. Depending on the propagation environment, the careful placement of antennas and nodes can lead to negative dependencies. The fading parameters related to geometry can be negatively dependent, e.g., in an opportunistic beamforming scheme with power control the signal strength of different beamforming vectors are negatively dependent. Finally, recent results on reconfigurable meta-surfaces \cite{Renzo2019} indicate the possibility of tuning the propagation environment and form the resulting wireless channels in a smart and flexible way. 

The work can be extended to non-monotone densities, i.e., to tail-monotone densities, e.g., Nakagami-m fading channels. We leave this as future work, because the bounds do not carry over to this case immediately \cite[Section 3.3]{Wang2013}. 

For practical applications, a finite blocklength is required. Especially in the context of \gls{urllc}, short blocklengths are needed to ensure a low latency. Therefore, it will be interesting for future work to extend the results to finite blocklengths, e.g., by using Polyanski's work on finite blocklength communication~\cite{Polyanskiy2010,Yang2014}.
\section*{Acknowledgment}
The authors would like to thank the anonymous reviewers for their technical and editorial comments, which significantly helped to improve the presentation of the results.

\appendix
\section{Appendix}
In the appendix, we state some derivations and relations which are needed in the main part.

Table~\ref{tab:formula-rayleigh} lists the different functions for the specific example of Rayleigh fading presented in Section~\ref{sec:rayleigh}.

\subsection{Expressions for \texorpdfstring{$-X$}{-X}}
If $X\sim F$ is distributed according to a monotone decreasing density $f(x)$, $-X$ has a monotone increasing density $f_{-}(x)=f(-x)$.
This yields the following relations
\begin{align}
	F_{-}(x) &= \int_{-\infty}^{x}f(-t)\diff{t}\\
	&= \int_{-x}^{\infty}f(a)\diff{a}\\
	&= 1 - F(-x)\,,
\end{align}
where $F(x)$ is the \gls{cdf} of $X$.
For the inverse $G_{-}$, it directly follows that
\begin{align}
	G_{-}(x) &= \inv{F_{-}}(x)\\
	&= \inv{\left(1-F(-x)\right)}\\
	&= -\inv{F}(1-x)\\
	&= -G(1-x)
\end{align}

We assume that $f(x)$ is a decreasing density function. Therefore, $H_a(x)$ is defined according to \eqref{eq:h-decreasing}. Since $f_{-}$ is increasing, $H^{-}_{a}(x)$ is defined according to \eqref{eq:h-increasing}, which can be rewritten using the above relations as follows
\begin{align}
H_{a}^{-}(x) &= G_{-}(a+x) + (n-1)G_{-}\left(1-(n-1)x\right)\\
&= -G(1-a-x) - (n-1) G(1 - 1 + (n-1)x)\\
&= -\Big(G(1-a-x) + (n-1) G((n-1)x)\Big)
\end{align}
using only the function $G$ of the original density $f$.

\subsection{Conditional Expectation}
The function $\phi$ defined in \eqref{eq:phi-decreasing} and \eqref{eq:phi-increasing} is calculated according to the conditional expectation $\expect{X|X>G(a)}$, if $c_n(a)=0$.
This conditional expectation is calculated as follows
\begin{align}
\expect{X|X > G(a)} &= \int_{-\infty}^{\infty}{x f_{X|X > G(a)}(x)}\diff{x}\\
&= \frac{1}{\Pr(X > G(a))}\int_{G(a)}^{\infty}x f(x) \diff{x}\\
&= \frac{1}{1-\underbrace{F(G(a))}_{a}}\int_{G(a)}^{\infty}x f(x) \diff{x}\,.
\end{align}
Applying the previously stated relations between $X$ and $-X$ gives the following expression
\begin{align}
\mathbb{E}_{-X\sim F}&\left[-X|-X > G_{-}(a)\right] = \int_{-\infty}^{\infty}{x f_{-X|-X > G_{-}(a)}(x)}\diff{x}\\
&= \frac{1}{\Pr(X > G(a))}\int_{G_{-}(a)}^{\infty}x f_{-}(x) \diff{x}\\
&= \frac{1}{1-\underbrace{F(G(a))}_{a}}\int_{G_{-}(a)}^{\infty}x f(-x) \diff{x}\\
&= \frac{1}{1-a}\int_{-G_{-}(a)}^{-\infty}t f(t) \diff{t}\\
&= \frac{-1}{1-a}\int_{-\infty}^{G(1-a)}t f(t) \diff{t}\\
&= \frac{-1}{1-a}\left(F(G(1-a))\cdot{}\expect[X\sim F]{X|X<G(1-a)}\right)\\
&= -\expect[X\sim F]{X|X<G(1-a)}
\end{align}%

\subsection{Function \texorpdfstring{$\phi$}{phi}}
For decreasing densities, the function $\phi$ is defined according to \eqref{eq:phi-decreasing}. Therefore, $\phi_{-}$ is defined according to \eqref{eq:phi-increasing}. Combining the above derivations yields the following expression
\begin{equation}
\phi_{-}(a) = \begin{cases}
-\left(G(1-a) + (n-1) G(0)\right) & \text{if } c^{-}_{n}(a) > 0\\
-n\expect{X|X<G(1-a)} & \text{if } c^{-}_{n}(a) = 0
\end{cases}\,,
\end{equation}
where $c_{n}^{-}(a)$ is defined according to \eqref{eq:cmin-increasing} for $f_{-}(x)$.
\subsection{Rayleigh Fading}\label{app:rayleigh-fading-example}
In the specific example of Rayleigh fading, we have $\abs{h_i}^2\sim\exp(\lambda_i)$. The different functions derived in the previous sections are listed in Table~\ref{tab:formula-rayleigh}.

\begin{table}[htb]
	\renewcommand{\arraystretch}{1.25}
	\centering
	\caption{Evaluated Functions for $X\sim\exp(\lambda)$}
	\begin{tabular}{cc}
		\toprule
		\textbf{Function} & \textbf{Expression}\\
		\midrule
		$f(x)$ & $\lambda\exp(-\lambda x)$\\
		$F(x)$ & $1 - \exp(-\lambda x)$\\
		$G(x)$ & $-\frac{\log(1-x)}{\lambda}$\\
		$H_a(x)$ & $-(n-1)\frac{\log(1-a-(n-1)x)}{\lambda} -\frac{\log(x)}{\lambda}$\\
		$\expect[X]{X|X>G(a)}$ & $\frac{1}{\lambda} + G(a)$\\
		$\phi(a)$ & $\displaystyle\begin{cases}
		H_{a}(c_{n}(a)) & \text{if } c_{n}(a) > 0\\
		n\frac{1-\log(1-a)}{\lambda} & \text{if } c_{n}(a) = 0
		\end{cases}$\\
		
		\midrule
		$f_{-}(x)$ & $\lambda\exp(\lambda x)$\\
		$F_{-}(x)$ & $\exp(\lambda x)$\\
		$G_{-}(x)$ & $\frac{\log(x)}{\lambda}$\\
		$H_a^{-}(x)$ & $\frac{\log(a+x)}{\lambda} + (n-1) \frac{\log(1-(n-1)x)}{\lambda}$\\
		$\expect[-X]{-X|-X>G_{-}(a)}$ & $\frac{a - a\log(a) - 1}{1-a}\cdot{}\frac{1}{\lambda}$\\
		$\phi_{-}(a)$ & $\displaystyle\begin{cases}
		H_{a}^{-}(0) & \text{if } c^{-}_{n}(a) > 0\\
		\frac{a - a\log(a) - 1}{1-a}\cdot{}\frac{1}{\lambda} & \text{if } c^{-}_{n}(a) = 0
		\end{cases}$\\
		\bottomrule
	\end{tabular}
	\label{tab:formula-rayleigh}
\end{table}
\subsection{Convex Minimization Problem}\label{app:minimization-problem}

In \cite[Thm.~3.4]{Wang2011}, the following solution to the minimization problem is proven for $X_i\sim F$ with a monotone density
\begin{equation}
\min_{X_1, \dots{}, X_n\sim F}\expect{g(X_1+\cdots{}+X_n)} = \expect[Q_{n}^{F}]{g(X_1 + \cdots{}+X_n)}\,,
\end{equation}
where $Q_{n}^{F}$ is the particular dependency structure constructed in \cite{Wang2011} and $g$ is a convex function.
The minimal value in the case of increasing densities is stated in \cite[Thm.~3.5]{Wang2011} as \eqref{eq:app-min-increasing-dens-wang} on top of the page.
\begin{figure*}[t]
	\begin{equation}\label{eq:app-min-increasing-dens-wang}
	\min_{X_1, \dots{}, X_n\sim F}\expect{g(X_1+\cdots{}+X_n)} = n\int_{0}^{c_n}g\left(H(x)\right)\diff{x} + (1-c_n)g\left(H(c_n)\right)\,.
	\end{equation}
\end{figure*}

In Lemma~\ref{lem:zero-out-perfect-csi-best}, we need the solution to this optimization problem for decreasing densities.
By combining the above with the proof of \cite[Prop.~3.1]{Wang2013}, this can be derived as follows.
Based on the properties of the dependency structure $(U_1, \dots{}, U_n)\sim Q_{n}^{F}$ for decreasing densities stated in \cite[Sec.~3.1]{Wang2013}, set $A_i=\left\{U_i\in[1-c, 1]\right\}$ and $B_i=\left\{U_i\in [0, (n-1)c]\right\}$. The variables $U_i\sim\mathcal{U}(0, 1)$ are independent uniformly distributed random variables, such that $X_i = G(U_i)$. Define $D_i=A_i\cup B_i$ and therefore $D_i^{\complement}=\left\{U_i\in((n-1)c, 1-c)\right\}$. Because of the structure of $Q_n^F$, the sum is a constant on the set $D_i^{\complement}$ equal to $H_0(c_n)$.
Based on the proof of \cite[Thm.~3.5]{Wang2011}, this can be used to derive the minimum value for decreasing densities
\begin{align}
&\min_{X_1, \dots{}, X_n\sim F}\expect{g(X_1+\cdots{}+X_n)}\\ &= \expect[Q_{n}^{F}]{g(G(U_1) + \cdots{}+G(U_n))}\\
\begin{split}
&= n\expect[Q_{n}^{F}]{g(G(U_1) + \cdots{}+G(U_n))\mathbf{1}_{A_1}} \\
&\qquad\quad+ \expect[Q_{n}^{F}]{g(G(U_1) + \cdots{}+G(U_n))\mathbf{1}_{D_1^{\complement}}}
\end{split}\\
&= n\expect[\mathcal{U}]{g(H_0(1-U_1))\mathbf{1}_{A_1}} + \expect[\mathcal{U}]{g(H_0(c_n))\mathbf{1}_{D_1^{\complement}}}\\
&= n\int_{1-c_n}^{1}g(H_0(1-x))\diff{x} + (1-c_n)g(H_0(c_n))\,.
\end{align}
All details and rigorous proofs needed for the above can be found in \cite{Wang2011,Wang2013}. Especially the proofs of \cite[Thm.~3.5]{Wang2011} and \cite[Prop.~3.1]{Wang2013} have been used here.
\printbibliography
\end{document}